\documentclass{amsart}
\usepackage[numbers,sort&compress]{natbib} 
\usepackage{verbatim}
\usepackage{lineno}
\begin{document}

\newcommand{\0}{\mathbf{0}}
\newcommand{\id}{\mbox{Id}}
\newcommand{\supp}{\mbox{supp}}
\newcommand{\I}{\mbox{int}}
\newcommand{\R}{\mathbf{R}}
\newcommand{\diag}{\mathbf{diag}}
\newcommand{\bR}{\partial\mathbf{R}}
\newcommand{\V}{\R^n}  
\newcommand{\C}{\mathbf{C}}  
\newcommand{\Cp}{\I\Sn }  
\newcommand{\Vi}{\R^{n_i}} 
\newcommand{\Ci}{\mathbf{C}_i} 
\newcommand{\Cip}{\I\Ci} 
\renewcommand{\H}{\mathbf{H}} 

\newtheorem{theorem}{Theorem}
\newtheorem{corollary}{Corollary}
\newtheorem{question}{Question}
\newtheorem{conjecture}{Conjecture}
\newtheorem{lemma}{Lemma}
\newtheorem{remark}{Remark}
\newtheorem{proposition}{Proposition}
\newcommand{\nz}{\hfill\break}
\newcommand{\lin}{\mbox{lin }}

\title{Robust permanence for interacting structured populations}
\author{Josef Hofbauer}
\author{Sebastian J. Schreiber}


\begin{abstract}
The dynamics of interacting structured populations can be  modeled by  $\frac{dx_i}{dt}= A_i (x)x_i$ where $x_i\in \R^{n_i}$, $x=(x_1,\dots,x_k)$,
and $A_i(x)$ are  matrices with non-negative off-diagonal entries. These models are permanent if there exists a positive global attractor and are robustly permanent if
 they  remain permanent following  perturbations of $A_i(x)$. Necessary and sufficient conditions for robust permanence are derived using dominant Lyapunov
 exponents $\lambda_i(\mu)$ of  the $A_i(x)$ with respect to
 invariant measures $\mu$.  The necessary
condition requires  $\max_i \lambda_i(\mu)>0$ for all
ergodic measures with support in the boundary of the non-negative cone. The
sufficient condition requires that the boundary  admits a Morse decomposition
such that  $\max_i \lambda_i(\mu)>0$
 for all invariant measures $\mu$ supported by a component of the Morse decomposition. When the Morse components are Axiom A, uniquely ergodic, or support all but one population,
 the necessary and sufficient conditions are equivalent.  Applications to spatial ecology, epidemiology, and  gene networks are given.
\end{abstract}

\maketitle

\centerline{\Large Appeared in \emph{Journal of Differential Equations}, 248, 1955-1971 (2010)}

\bibliographystyle{unsrtnat}

\section{Introduction}

A fundamental issue in population biology is what are the minimal conditions to ensure the long-term survivorship for  interacting populations whether they be viral particles, bio-chemicals, plants, or animals. When these conditions are met the interacting populations are said to persist or coexist. Since the pioneering work  of
Lotka and Volterra on competitive and predator--prey
interactions, Thompson,
Nicholson, and Bailey on host--parasite interactions, and
Kermack and McKendrick on disease outbreaks,
nonlinear difference and differential equations have been
used to understand conditions for population persistence~\cite{lotka-25,volterra-26,nicholson-bailey-35,thompson-24,kermack-mckendrick-27}. One particularly important form of persistence for deterministic models is \emph{permanence} or \emph{uniform persistence} which corresponds to the existence of a global attractor bounded away from extinction of one or more species. When such an attractor exists, interacting populations are able to recover from ``vigorous shake ups'' of the population state. Permanence has been characterized from a topological perspective~\cite{butler-waltman-86,hofbauer-so-89,garay-89} and with average Lyapunov functions~\cite{hofbauer-81,hutson-84,hutson-88,hofbauer-sigmund-98} for differential equation models of the form
\begin{equation}\label{eq:scalar}
\frac{dx_i}{dt}= x_i f_i(x) \qquad i=1,\dots,k
\end{equation}
where $x_i$ is the abundance of population $i$, $x=(x_1,\dots,x_k)$, and $f_i(x)$ is the per-capita growth of population $i$.

An important extension of the concept of permanence is \emph{robust permanence} \cite{hutson-schmitt-92,jde-00,hirsch-smith-zhao-01,garay-hofbauer-03} which requires that  \eqref{eq:scalar} remain permanent
after sufficiently small perturbations of the per-capita growth rates $f_i$. The importance of robust permanence stems from the fact that all models are approximations to reality. Consequently, if nearby
models (e.g. more realistic models) are not permanent despite the focal model being permanent, then one
can draw few (if any!) conclusions about the persistence of the
biological system being approximated by the model. One can view robust permanence as one crude form of structural stability for population models. For this perspective, it is not unexpected that there are permanent systems that cannot be approximated by robustly permanent systems \cite{nonlinearity-04}. In \cite{jde-00,garay-hofbauer-03}, criteria for robust permanence were developed with respect to the average per-capita growth rates $\int f_i(x) \,d\mu(x)$ with respect to invariant probability measures $\mu$ supported on the boundary of the positive orthant (i.e. where one or more populations are extinct). Roughly, these criteria for robust permanence require that the average per-capita growth rate (i.e. invasion rate) is positive for some missing species.

While \eqref{eq:scalar} can account for many types of population interactions, it assumes that all individuals within a population are exactly the same. However, theoretical biologists have long recognized that different individuals within a population may be in different states (e.g. different sizes or ages, living in different spatial locations) and these differences can have important consequences for population dynamics \cite{caswell-01,kon-etal-04,kon-05,kon-iwasa-07,salceanu-smith-09,deleenheer-etal-09a,deleenheer-etal-09b}. To account for how these differences influence persistence of interacting populations,  we develop criteria for robust permanence for the dynamics of $k$ interacting structured
populations. Roughly, these dynamics correspond to replacing $x_i$ in \eqref{eq:scalar} by vectors and $f_i(x)$ by matrices.   In section~\ref{sec:models}, we describe these models in greater detail and review some basic concepts from dynamical systems theory. In section~\ref{sect:invasion}, we introduce the structured analog of average per-capita growth rates and define the concept of robustly
unsaturated invariant sets. Necessary and sufficient
conditions for an invariant set to be robustly unsaturated are proven. For Axiom A or uniquely ergodic invariant sets, the necessary and sufficient criteria are shown to be equivalent. In
section~\ref{sect:permanence}, we use Morse decompositions and the
criteria for robustly unsaturated invariant sets to develop
necessary and sufficient conditions for robust permanence. In
sections~\ref{sect:patch} and~\ref{sect:applications}, we provide several applications of
our results to spatially structured ecological models, structured epidemiological models, and models of  gene networks.

\section{Models and assumptions\label{sec:models}}

 Let $x_i$ denote the state of the $i$-th population
that lies in the non-negative cone $\Ci$ of $\Vi$.  Define
$x=(x_1,\dots,x_k)\in \C$ to be the non-negative cone of $\V$ where
$n=\sum_{i=1}^k n_i$. Let $\Ci^+$  denote positive cone $\{x_i\in \Ci:
\prod_j x_i^j>0 \}$ for population $i$ and $\C^+=\prod_i \Ci^+$, the positive cone for the interacting populations. If $x\mapsto A_i(x)$ is a
map into $n_i\times n_i$ matrices that describes the growth of
population $i$, then the dynamics of the interacting populations are
given by
\begin{equation}\label{eq:main}
\frac{dx_i}{dt}=A_i(x)x_i \qquad i=1,\dots, k \end{equation}  Let
$x.t$ denote the solution of (\ref{eq:main}) with the initial
condition $x$, and more generally the semiflow generated by \eqref{eq:main}.
Before stating our assumption on (\ref{eq:main}),
recall a few definitions from dynamical systems. Given sets
$I\subseteq \R$ and $K\subseteq \C$, let $K.I=\{x.t:t\in I,x\in
K\}$. A set $K\subseteq \C$ is \emph{invariant} if $K.t=K$ for all
$t > 0$, and \emph{forward invariant} if $K.t \subset K$ for all $t > 0$.
The \emph{omega limit set} of a set $K\subseteq \C$ equals
$\omega(K)=\bigcap_{t\ge 0}\overline{K.[t,\infty)}$. The \emph{ alpha
limit set} of $K\subseteq \C$ is $\alpha(K)=\bigcap_{t\le
0}\overline{K.(-\infty,t]}$. Given a forward invariant set $K$, $B\subset
K$ is an \emph{attractor for the semiflow $x.t$ restricted to $K$} provided there
exists an open neighborhood $U\subseteq K$ of $B$ such that
$\omega(U)=B$.  The \emph{stable set of a compact invariant set $K$}
is defined by
   \[
    W^s(K)=\{x\in \C: \omega (x)\neq \emptyset\hbox{ and }
    \omega(x)\subseteq K\}.
    \]
The semiflow generated by \eqref{eq:main} is \emph{dissipative} if there exists an attractor $B$ with
$W^s(B)=\C$.

Throughout this paper, we make the following assumptions:
\begin{description}\item [A1]  $x\mapsto A_i(x)$ is continuous
\item [A2] $A_i(x)$ is irreducible and has non-negative off-diagonal entries.
\item [A3] $x.t$ is defined for all $t\ge 0$
\item [A4] $x.t$ is dissipative with global attractor $\Gamma(A)$.
\end{description}
\textbf{A1} is a basic regularity assumption. The non-negativity of
off-diagonal entries in \textbf{A2} implies that there are no negative
feedbacks between individuals of different states in population $i$.
This assumption is meet for many types of structured models as
discussed in sections~\ref{sect:patch} and \ref{sect:applications}.
 The irreducibility assumption of \textbf{A2} is
generically meet and implies that all individuals within a
population can pass through all states. In Remark~\ref{remark}, we discuss how this irreducibility assumption can be relaxed.
 \textbf{A3} ensures that the
population dynamics are defined for all future time. \textbf{A4} requires
that population densities/abundances eventually are uniformly
bounded, a condition that should be met for an biologically realistic model.

We say that the semiflow of \eqref{eq:main} is
\emph{permanent} if this semiflow is dissipative and there is a positive
attractor $B\subset \C^+$ such that $W^s(B) \subset \C^+$.
We will denote this positive attractor by $\Lambda(A)$.
Permanence  ensures that
populations can recover from rare large perturbations and allows for
a diversity of dynamical behaviors.

For $K\subset \C$ and $\delta> 0$, define the
\emph{$\delta$-neighborhood of $K$} as
\[
N_\delta(K)=\{x\in \C: |x-y|<\delta \mbox{ for some }y\in K\}.
\]
We define a \emph{ $\delta$-perturbation of (\ref{eq:main})} to be a
system of the form
\[
\frac{dx}{dt}=\tilde A(x)x
\]
that satisfies assumptions \textbf{A1}--\textbf{A4},
$\Gamma(\tilde A)\subset N_\delta(\Gamma(A))$, and  $\|A(x)-\tilde
A(x)\|\le \delta $ for all $x\in N_\delta(\Gamma(A))$.
(\ref{eq:main}) is \emph{robustly permanent} if there exist
$\delta>0$ and $\epsilon>0$ such that all $\delta$-perturbations of (\ref{eq:main})
are permanent, and $d(\Lambda(\tilde A), \C \setminus \C^+) > \epsilon$ for all $\delta$-perturbations i.e., there is a common/uniform region of repulsion around the boundary.

\section{Invasion of compact sets\label{sect:invasion}}
We begin by studying the linear skew product flows on
$\Gamma(A)\times \Vi$ defined by $(x,y).t=(x.t,B_i(t,x)y)$ where
$Y(t)=B_i(t,x)$ is the solution to $\frac{dY}{dt} (t)= A_i(x.t) Y(t)$ with
$Y(0)$ equal to the identity matrix. Our assumption that $A_i$ is
irreducible with non-negative off diagonal entries implies that
$B_i(t,x)\Ci\subset \Ci^+$ for all $x$ and $t>0$ (see, e.g., \cite{smith-95}). A result
of Ruelle~\cite[Prop.3.2]{ruelle-79} implies that there
exist continuous maps $u_i,v_i:\Gamma(A)\to \Ci^+$ with
$|u_i(x)|=|v_i(x)|=1$ such that
\begin{itemize}
\item the line bundle $E_i(x)$ spanned by $u_i(x)$ is invariant i.e.
$E_i(x.t)=B_i(x,t)E_i(x)$ for all $t\ge 0$.
\item the vector bundle $F_i(x)$ perpendicular to $v_i(x)$ is invariant
i.e. $F_i(x.t)=B_i(x,t)F_i(x)$ for all $t\ge 0$.
\item there exist constants $\alpha>0$ and $\beta>0$ such that
\begin{equation}\label{perron}
\|B_i(t,x)|F_i(x)\| \le \alpha \exp(-\beta t) \|B_i(t,x)|E_i(x)\|
\end{equation}
for all $x\in K$ and $t\ge 0$.
\end{itemize}

\subsection{Invasion rates}
Given any $x\in \Gamma(A)$, define \emph{the invasion rate of species $i$
at population state $x$} as
\[
\lambda_i(x)=\limsup_{t\to\infty} \frac{1}{t} \ln \|B_i(t,x)\|
\]
When $x_i=0$, $\lambda_i(x)$ provides an upper bound to the rate of
growth of population $i$ when introduced at infinitesimally small
densities. Important properties of this invasion rate are summarized
in the following two propositions. For instance, the second proposition implies
if $(x.t)_i$ stays bounded away from zero (i.e. population $i$
persists), then $\lambda_i(x)=0$. Alternatively, if $\lambda_i(x)<0$, then
population $i$ is doomed to extinction.


\begin{proposition}\label{prop:invasion rate}
$\lambda_i(x)$ satisfies the following properties:
\begin{itemize}
\item $\lambda_i(x)=\limsup_{t\to\infty} \frac{1}{t} \ln |B_i(t,x)v|$ for any
$v>0$ in $\Vi$, and
\item $\lambda_i(x)=\limsup_{t\to\infty} \frac{1}{t}\int_0^t
\langle A_i(x.s)u_i(x.s),u_i(x.s) \rangle \,ds$. \end{itemize}
\end{proposition}

\begin{proof}
To prove the first property, we first show that
\[\lambda_i(x)=\limsup_{t\to\infty}\frac{1}{t}\ln
|B_i(t,x)u_i(x)|.\] To this end, let $v\in \Vi$ be any non-zero vector. Since
$\Vi=E_i(x)\oplus F_i(x)$, we can write $v=a u_i(x)+w$ with $a\in
\R$ and $w\in F_i(x)$. Equation (\ref{perron}) implies
\begin{eqnarray*}
|B_i(t,x)v| &\le& a|B_i(t,x)u_i(x)|+|B_i(t,x)w|\\
& \le&  |B_i(t,x)u_i(x)|\left(a+\alpha\exp(-\beta t)|w|\right).
\end{eqnarray*}
Hence,
\[
\limsup_{t\to\infty}\frac{1}{t}\ln |B_i(t,x)v| \le
\limsup_{t\to\infty}\frac{1}{t}\ln |B_i(t,x)u_i(x)|\le
\lambda_i(x)
\]
for all non-zero vectors $v\in \Vi$. Since $\|B_i(t,x)\|=\sup_{|v|=1}
|B_i(t,x)v|$, this inequality implies that
$\lambda_i(x)=\limsup_{t\to\infty}\frac{1}{t} \ln
|B_i(t,x)u_i(x)|$. Now let $v\in \Ci^+$. Then, we can
write $v=au_i(x)+w$ with $a>0$ and $w\in F_i$. Equation
(\ref{perron}) implies
\begin{eqnarray*}
|B_i(t,x)v|&\ge & a|B_i(t,x)u_i(x)|-|B_i(t,x)w|\\
&\ge & |B_i(t,x)u_i(x)|(a-\alpha\exp(-\beta t)|w|)
\end{eqnarray*}
Since $a>0$, this inequality implies that
\[
\lambda_i(x)\ge \limsup_{t\to\infty}\frac{1}{t}\ln |B_i(t,x)v| \ge
\limsup_{t\to\infty}\frac{1}{t}\ln |B_i(t,x)u_i(x)| =\lambda_i(x)
\]

To prove the second assertion, let $b_i(t,x)=\ln |B_i(t,x)u_i(x)|$.
Invariance of $u_i(x)$ implies that $b_i(t,x)$ is additive:
\begin{eqnarray*}
b_i(t+s,x)&=&\ln |B_i(t+s,x)u_i(x)| =\ln
|B_i(s,x.t)B_i(t,x)u_i(x)|\\
&=& \ln |B_i(s,x.t) u_i(x.t)||B_i(t,x)u_i(x)|\\
&=&\ln
|B_i(s,x.t)u_i(x.t)|+\ln |B_i(t,x)u_i(x)|\\
&=& b_i(s,x.t)+b_i(t,x) \end{eqnarray*}
Additivity of $b_i(t,x)$ and
the fact that $b_i(0,x)=0$ implies
\begin{eqnarray*}
\frac{d}{dt} b_i(t,x)&=& \lim_{s\to 0} \frac{b_i(t+s,x)-b_i(t,x)}{s}\\
&=& \lim_{s\to 0} \frac{b_i(s,x.t)}{s}= \frac{d}{ds}\Big|_{s=0} b_i(s,x.t)\\
&=& \frac{\langle\frac{d}{ds} B_i(s, x.t)u_i(x.t),B_i(s,x.t)u_i(x.t)\rangle}{
|B_i(s,x.t)u_i(x.t)|^2}\Big|_{s=0}\\
&=& \langle A_i(x.t)u_i(x.t),u_i(x.t)\rangle
\end{eqnarray*}
The Fundamental Theorem of Calculus implies
\[
b_i(t,x)=\int_0^t \langle A_i(x.s)u_i(x.s),u_i(x.s)\rangle \,ds
\]
and the second assertion follows. \end{proof}

\begin{proposition}\label{prop:invasion rate2}
For the solutions of  \eqref{eq:main} we have:
\begin{itemize}
\item if $x_i>0$, then $\lambda_i(x)\le 0$.
\item if $\lambda_i(x)<0$, then $\lim_{t\to\infty} (x.t)_i=0$
\item if $\limsup_{t\to\infty} |(x.t)_i|>0$, then
$\lambda_i(x)=0$.
\end{itemize}
\end{proposition}

\begin{proof}
First, assume that $x$ satisfies $x_i>0$.
Since the semiflow is dissipative, there exists $\gamma>0$ such that
$|x.t| \le \gamma$ for all $t\ge 0$. Proposition~\ref{prop:invasion rate} and the
definition of the skew product flow imply
\begin{eqnarray*}
\lambda_i(x)&=& \limsup_{t\to\infty} \frac{1}{t} \ln |B_i(t,x)x_i|=\limsup_{t\to\infty} \frac{1}{t} \ln |(x.t)_i|\\
&\le&\limsup_{t\to\infty} \frac{\ln \gamma}{t} =0
\end{eqnarray*}

To prove the second  assertion, assume that $\lambda_i(x)<0$. If
$x_i=0$, then the invariance of the faces of $\C$ imply that
$(x.t)_i=0$ for all $t\ge 0$. Alternatively, if $x_i>0$, then Proposition~\ref{prop:invasion rate} and the definition of the skew product flow imply
\begin{eqnarray*}
\limsup_{t\to\infty} \frac{1}{t} \ln |(x.t)_i| &=& \lambda_i(x)<0.
\end{eqnarray*}

The final assertion follows from the first two assertions.
\end{proof}

\subsection{Invariant measures}
We review some definitions from ergodic theory.  Given a Borel
probability measure $\mu$ on $\C$, the {\em support of $\mu$},
denoted $\supp(\mu)$, is the smallest closed set whose complement
has measure zero. A Borel probability measure $\mu$ is called \emph{
invariant} for \eqref{eq:main} provided that $\int h(x)\,d\mu(x)= \int h(x.t)\,d\mu(x)$ for all
$t\ge 0$ and for all bounded continuous functions $h:\C\to\R$. An invariant
measure $\mu$ is called \emph{ergodic} provided that $\mu(B)=0$ or
$1$ for any invariant Borel set $B$.

For an invariant measure $\mu$ define the \emph{invasion rate of
species $i$ with respect to $\mu$} as
\[
\lambda_i(\mu)=\int_\C \langle A_i(x)u_i(x),u_i(x) \rangle \, d\mu(x)
\]

\begin{proposition}\label{prop:invasion2}
Let $\mu$ be an invariant measure. $\lambda_i(\mu)$ satisfies the
following properties:
\begin{itemize}
\item $\lambda_i(x)=\lim_{t\to\infty} \frac{1}{t}\int_0^t
\langle A_i(x.s)u_i(x.s),u_i(x.s) \rangle ds$ exists  for $\mu$-almost every $x$.
Moreover, if $\mu$ is ergodic, $\lambda_i(x)=\lambda_i(\mu)$ for
$\mu$-almost every $x$.
\item if $\mu$ is ergodic, then there exists $I\subset \{1,\dots, k\}$ such that
$\mu(\prod_{i\in I} \Ci^+)=1$ and $\lambda_i(\mu)=0$ for all $i\in
I$.
\end{itemize}
\end{proposition}

\begin{proof}
Let $\mu$ be an invariant measure. Define
$h_i(x)=\langle A_i(x)u_i(x),u_i(x)\rangle$. The Birkhoff ergodic theorem and Proposition~\ref{prop:invasion rate} imply
\[
\lambda_i(x)=\lim_{t\to\infty} \frac{1}{t} \int_0^t h_i(x.s)\,ds
\]
exists for $\mu$-almost every $x$. Moreover, $\lambda_i(x)=\lambda_i(\mu)$ $\mu$-almost surely if $\mu$
is ergodic.

Assume $\mu$ is ergodic. By ergodicity and invariance of the faces
of $\C$, there exists $I\subset \{1,\dots, k\}$ such that
$\mu(\prod_{i\in I}\Ci^+)=1$. Let $K\subset \prod_{i\in I}\Ci^+$ be
a compact set such that $\mu(K)>0$. The Poincar\'{e} recurrence
theorem and the Birkhoff ergodic theorem imply that there is $x\in \prod_{i\in I} \Ci^+$ such that
$x.t_n\in K$ for some $t_n\uparrow \infty$ and
$\lambda_i(x)=\lambda_i(\mu)$ for all $i$. The third assertion of Proposition~\ref{prop:invasion rate2} implies
$\lambda_i(\mu)=0$ for all $i\in I$.
\end{proof}

\subsection{Robustly unsaturated sets}

Let $K\subset \partial \C$ be a compact isolated invariant set for
the flow $x.t$ restricted to $\partial \C$. $K$ is
\emph{unsaturated} if  $W^s(K)\subset \partial \C$ and $K$ is isolated for
$x.t$. If $K$ is not unsaturated, then $K$ is \emph{saturated}. $K$
is \emph{robustly unsaturated for (\ref{eq:main})} if there exists
$\delta>0$ such that the continuation of $K$ for any
$\delta$-perturbation of (\ref{eq:main}) is unsaturated.

\begin{theorem}\label{thm:unsaturated}
Let $K$ be a compact isolated invariant set for $x.t$ restricted to
$\partial \C$. If one of the the following equivalent conditions
hold
\begin{itemize}
\item for all invariant measures $\mu$ supported by $K$
\begin{equation}\label{eq:condition}
\max_{1\le i\le k} \lambda_i(\mu)>0
\end{equation}
\item there exist $p_i>0$ such that
\begin{equation}\label{eq:condition2}
\sum_{1\le i\le k} p_i \lambda_i(\mu)
>0 \end{equation} for all ergodic probability measures supported by $K$.
\end{itemize}
then $K$ is robustly unsaturated for (\ref{eq:main}). Alternatively,
if $x\mapsto A(x)$ is  twice continuously differentiable and $K$ is robustly unsaturated, then
(\ref{eq:condition}) holds for all ergodic measures $\mu$ supported
by $K$. \end{theorem}

\begin{remark}\label{remark}
For some applications (e.g., the disease model considered in section~\ref{sect:applications:disease}), it useful to relax the irreducibility assumption \textbf{A2}. For instance, if there exists an open neighborhood $U$ of $K$ such that for each $i$, $A_i(x)$ has a fixed off-diagonal sign pattern for all $x\in U$, then $A_i(x)$ can be decomposed into a finite number, say $m_i$, of irreducible components. For each of these irreducible components, one can define $\lambda_i^j (\mu) = \int_\C \langle A_i^j(x) u_i^j(x),u_i^j(x)\rangle \,d\mu(x)$ where $A_i^j(x)$ is the submatrix of $A_i(x)$ corresponding to the $j$-th irreducible component of $A_i(x)$ and $u_i^j(x)$ is the continuous invariant subbundle for the irreducible linear cocycle determined by $A_i^j(x)$. If we define $\lambda_i(\mu)=\max_{1\le j\le m_i} \lambda_i^j (\mu)$, then all of the assertions of Theorem~\ref{thm:unsaturated} still hold.  
\end{remark}

\begin{proof}

To prove the first assertion of the theorem, we need the following
lemma. We call an invariant measure $\mu$ for (\ref{eq:main})
\emph{saturated} if $\lambda_i(\mu)\le 0$ for all $i$.

\begin{lemma}\label{lemma:sat} Let $\delta_n$ be a non-negative sequence that converges to zero as $n\rightarrow \infty$. If $\mu_n$ are saturated
invariant measures for $\delta_n$-perturbations of (\ref{eq:main}),
then the weak* limit points of $\{\mu_n\}_{n=1}^\infty$ is a
non-empty set consisting of saturated invariant measures for
(\ref{eq:main}).
\end{lemma}

\begin{proof}
Let $\delta_n$ be a non-negative sequence that converges to zero.
Let $\mu_n$ be saturated invariant measures for $\delta_n$
perturbations, $\dot x= A^n(x)x$, of $\dot x=A(x)x$. By weak*
compactness of Borel probability measures supported on
$N_1(\Gamma(A))$, there exist weak* limit points of $\mu_n$. Let
$\mu$ be such a weak* limit point. Since $\mu_n$ are supported by
$N_1(\Gamma(A))$ for all $n$ sufficiently large, $\mu$ is supported by
$N_1(\Gamma(A))$. To verify that $\mu$ is an invariant measure for
(\ref{eq:main}), let $h:\C\to \R$ be a bounded continuous function. Let $t>0$
and $\epsilon>0$ be given. Choose $n$ sufficiently large so that
$|\int_\C h(x)\,d\mu(x)-\int_\C h(x)\,d\mu_n(x)|\le \epsilon$,
$|\int_\C h(x.t)\,d\mu(x)-\int_\C h(x.t)\,d\mu_n(x)|\le \epsilon$,
and $|h(x.t)-h(x^n.t)|\le \epsilon$ for all $x\in N_1(\Gamma(A))$.
Then $|\int_\C h(x.t)- h(x)\,d\mu(x)|$ is
\begin{eqnarray*}
&\le &|\int_\C h(x.t)\,d\mu(x)-\int_\C
h(x.t)\,d\mu_n(x)|+|\int_\C
h(x.t)\,d\mu_n(x)-\int_\C h(x) \,d\mu(x)|\\
&\le & \epsilon +\int_\C |h(x.t)- h(x^n.t)|\,d\mu_n(x)+|\int_\C
h(x^n.t)\,d\mu_n(x)-\int_\C h(x)
\,d\mu(x)|\\
&\le &  2\epsilon+|\int_\C h(x^n.t)-h(x)\,d\mu_n(x)|+ |\int_\C
h(x)\,d\mu_n(x)- \int_\C h(x) \,d\mu(x)|\\
& \le &3\epsilon
\end{eqnarray*}
where the last line follows from the invariance of $\mu_n$ for
$x^n.t$. Since $\epsilon>0$ and $t>0$ are arbitrary, we have
$\int_\C h(x.t)\,d\mu(x)=\int_\C h(x)\,d\mu(x)$ for all $t>0$ and
all bounded continuous functions $h:\C\to \R$. It follows that $\mu$ is an
invariant measure for (\ref{eq:main}). To see that $\mu$ is
saturated, define $h_i^n(x)=\langle A^n_i(x)u_i^n(x),u_i^n(x)\rangle$ where
$u_i^n(x)$ spans the invariant one-dimensional bundle given by
\cite[Prop.3.2]{ruelle-79} for $\frac{dx}{dt}=A^n(x)x$. Since $A^n_i(x)\to A_i(x)$ uniformly for $x\in
N_1(\Lambda(A))$ as $n\to \infty$, \cite[Prop.3.2]{ruelle-79} implies
that $u_i^n(x)\to u_i(x)$ converges uniformly for $x\in
N_1(\Lambda(A))$. Given $\epsilon>0$, choose $n$ sufficiently large
so that $|h_i^n(x)-h_i(x)|\le \epsilon$ for all $x\in
N_1(\Lambda(A))$ and  $|\int_\C h_i^n(x)\,d\mu(x)-\int_\C
h_i^n(x)\,d\mu_n(x)|\le \epsilon$. Then
\begin{eqnarray*}
\lambda_i(\mu)&\le & \int_\C |h_i(x)-h_i^n(x)|\,d\mu(x)+|\int_\C
h_i^n(x)\,d\mu(x)-\int_\C h_i^n(x)\,d\mu_n(x)|\\&&+\int_\C
h_i^n(x)\,d\mu_n(x)\\
&\le& 2 \epsilon
\end{eqnarray*}
Since $\epsilon>0$ is arbitrary, $\lambda_i(\mu)\le 0$ for all $i$ and
$\mu$ is saturated.
\end{proof}

Using this lemma, we prove that if $K$ is saturated, then there
exists a saturated invariant measure $\mu$ supported by $K$. Assume
$K$ is saturated. Work of Hofbauer and So~\cite[Thm.2.1]{hofbauer-so-89} implies that either
there exists $y\in W^s(K)\cap \C^+$ or $K$ is not isolated for the
unrestricted flow $x.t$. If there exists $y\in W^s(K)\cap \C^+$,
then for all $t>0$ define $\nu_t=\frac{1}{t}\int_0^t
\delta_{y.s}\,ds$ where $\delta_{y.s}$ denotes a Dirac measure based
at the point $y.s$. Dissipativeness of (\ref{eq:main}) and weak*
compactness of the Borel probability measures supported by
$\Lambda(A)$ imply there exist $t_k\to \infty$ such that the
sequence $\nu_{t_k}$ converges in the weak* topology to a Borel
probability measure $\mu$ with support in $K$. A standard argument
implies that $\mu$ is $x.t$ invariant. Define $h_i:\C\to\R$ by
$h_i(x)=\langle A_i(x)u_i(x),u_i(x)\rangle $. Proposition~\ref{prop:invasion rate2}
and weak* convergence
imply that
\begin{eqnarray*}
0\ge \lambda_i(y)&=&\limsup_{t\to\infty}\frac{1}{t}\int_0^t h_i(y.s)\,ds\\
&=&\limsup_{t\to\infty} \int_\C h_i\,d\nu_t \\
&\ge & \limsup_{k\to\infty} \int_\C h_i\,d\nu_{t_k} = \lambda_i(\mu)
\end{eqnarray*}
for all $i$. Alternatively, suppose that $K$ is not isolated for
the semiflow. Then there exists a sequence of positive $\omega$-limit sets
that accumulate on $K$. Let $\mu_n$ be a sequence of ergodic
probability measures supported by these $\omega$-limit sets.
Propositions~\ref{prop:invasion rate} and \ref{prop:invasion2} imply
that $\mu_n$ are saturated for all $n$. Applying
Lemma~\ref{lemma:sat} with $\delta_n=0$ for all $n$ implies that
there exists a saturated invariant measure $\mu$ supported by $K$.
Hence, we have shown that if $K$ is saturated, then there exists a
saturated invariant probability measure $\mu$ supported by $K$.
Equivalently, (\ref{eq:condition}) holding for all invariant
probability measures with support in $K$ implies that $K$ is
unsaturated.

Next, we show that if $K$ is not robustly saturated, then there
exists a saturated invariant measure supported by $K$. Indeed,
suppose $K$ is not robustly saturated. Then there exists a
non-negative sequence $\delta_n$ converging to zero and a sequence
of saturated measures $\mu_n$ for $\delta_n$-perturbations of
(\ref{eq:main}) with support in the continuation of $K$. Let $\mu$
be a weak* limit point of $\{\mu_n\}_{n=1}^\infty$.
Lemma~\ref{lemma:sat} implies that $\mu$ is saturated. Moreover,
since the continuation of $K$ converges to $K$ as $\delta_n\to 0$,
$\mu$ is a saturated invariant measure for (\ref{eq:main}) supported
by $K$. Hence, we have shown that if $K$ is not robustly saturated,
then there exists a saturated invariant measure supported by $K$.
Equivalently, if (\ref{eq:condition}) holds for all invariant
measures with support in $K$, then $K$ is robustly unsaturated.
 
To see the equivalence of the conditions given by \eqref{eq:condition}  and \eqref{eq:condition2},
let $\Delta=\{p\in \R^k_+: \sum_i p_i =1\}$ and notice that
\[
\min_\mu \max_i \lambda_i (\mu) = \min_\mu\max_{p\in \Delta} \sum_i p_i \lambda_i(\mu)
\]
where the minimum is taken over invariant probability measures $\mu$ with support in $K$.
The Minimax theorem (see, e.g., \cite{simmons-98}) implies that
\begin{equation}\label{eq:minmax}
\min_\mu \max_i \lambda_i (\mu) = \max_{p\in \Delta} \min_\mu \sum_i p_i \lambda_i(\mu)
\end{equation}
where the minimum is taken over invariant probability measures $\mu$ with support in $K$. Since $\min_\mu \sum_i p_i \lambda_i(\mu)$
is attained at an ergodic probability measure with support in $K$, the equivalence of the conditions
given by \eqref{eq:condition} and \eqref{eq:condition2} is established.

To prove the final assertion of the Theorem, assume $x\mapsto A(x)$
is twice continuously differentiable and there exists a saturated ergodic measure $\mu$
supported by $K$. Proposition~\ref{prop:invasion2} implies that
there exists $I\subset \{1,\dots,k\}$ such that $\mu(\prod_{i\in I}
\Ci^+)=1$ and $\lambda_i(\mu)=0$ for all $i\in I$. Since $K\subset
\partial \C$, $\{1,\dots,k\}\setminus I$ is non-empty. We will show
that (\ref{eq:main}) is not robustly permanent by proving that for
all $\delta>0$ there exists a $\delta$-perturbation of
(\ref{eq:main}) for which the continuation of $K$ is saturated. Let
$\delta>0$ be given. Choose $V \subset N_\delta(\Gamma(A))$ to be a
compact neighborhood of $\Gamma(A)$ such that $V.t\subset \I V$ for
all $t>0$. Let $\eta>0$ be such that $W=:N_\eta(\Gamma(A))\subset
V$. Let $\rho:\C\to[0,1]$ be a $C^\infty$ function such that
$\rho(x)=1$ for all $x\in \Gamma(A)$ and $\rho(x)=0$ for all $x\in
\C\backslash W$.Define $\tilde A=(\tilde A_1,\dots,\tilde A_k)$ by
\[
    \tilde A_i(x)= \left\{
    \begin{array}{ll}
    A_i(x)&\mbox{if }i\in I\\
    A_i(x)-\frac{\delta}{2}\id \rho(x)&\mbox{if }i \notin I.
    \end{array}
    \right.
  \]
where $\id$ denotes the identity matrix of appropriate dimension.
Let $\widetilde {x.t}$ denote the semiflow of $\dot x = \tilde A(x)x$.
Since $\widetilde {x.t}=x.t$ whenever $x.[0,t]\in \C\setminus W$, it
follows that $\widetilde{V.t}\subset \I V$ for all $t>0$. Hence,
$\Gamma(\tilde A) \subset V \subset N_\delta(\Gamma(A)))$. We also
have $\|A(x)-\tilde A(x)\|\le \frac{\delta}{2}$. Therefore, $\dot x
= \tilde A(x)x$ is a $\delta$-perturbation of (\ref{eq:main}). By
construction, $x.t=\widetilde{x.t}$ for all $x\in K$ and $t\ge 0$.
Hence, $\mu$ is ergodic for the semiflow of $\dot x = \tilde A(x)x$ and $\lambda_i(\mu)\le
-\frac{\delta}{2}$ for this semiflow and $i\in \{1,\dots,
k\}\setminus I$. Let ${\mathcal L}$ and $O$ be the Lyapunov
exponents and Oseledec regular points supported by $\mu$ for $\dot x= \tilde A(x)x$ (see, e.g., \cite{pugh-shub-89} for definitions).
At each point $x\in O$, the splitting of $\R^n$ determines three
subspaces: the stable subspace
$E^s(x)$, the center subspace
$E^c(x)$ and the unstable subspace
$E^u(x)$.
Proposition~\ref{prop:invasion2} and our choice of $\tilde A$ imply
that $E^s(x)\cap \C^+\neq \emptyset $. The Pesin stable manifold
theorem \cite[Corollaries 3.17 and 3.18]{pugh-shub-89} implies that
tangent to $E^s(x)$, $E^c(x)$ and $E^u(x)$ are locally $\widetilde
{x.t}$-invariant families of $C^1$ discs ${\mathcal W}^s_x$,
${\mathcal W}^c_x$ and ${\mathcal W}^u_x$ corresponding to the
stable, center and unstable manifolds. The family of stable
manifolds ${\mathcal W}^s_x$ is contained in $W^s(\tilde A,\supp(
\mu))$. Since $E^s(x)\cap \C^+\neq \emptyset$, ${\mathcal W}^s_x
\cap \C^+\neq \emptyset$ for some $x\in K$. Consequently, $K$ is
saturated for $\dot x=\tilde A(x)x$.
\end{proof}

The difference between the sufficient and necessary condition is
that for the sufficient condition, (\ref{eq:condition}) has to hold
for \emph{all invariant} measures supported by $K$, while for the
necessary condition, (\ref{eq:condition}) has to hold only for
\emph{ergodic} measures supported by $K$. Since the invariant
measures lie in the convex hull of the ergodic measures, the
sufficient condition can be  more restrictive than the necessary
condition.

For three classes of invariant sets, the necessary and
sufficient conditions coincide. Recall, a compact invariant set $K$
for \eqref{eq:main} is \emph{uniquely ergodic} if $K$ only supports one
invariant measure. Recall a compact invariant set $K$ is \emph{Axiom
A} if the flow of \eqref{eq:main} restricted to $K$ is transitive and  hyperbolic (see, e.g., \cite{shub-87} for definitions).

\begin{corollary}\label{cor:axiomA}
Assume $x\mapsto A(x)$ is  twice continuously differentiable. If $K\subset \partial \C$ is a compact invariant set and either it is
uniquely ergodic or Axiom A, then $K$ is robustly unsaturated if and
only if (\ref{eq:condition}) holds for all ergodic measures $\mu$
with support in $K$.
\end{corollary}

\begin{proof}
If $K$ is uniquely ergodic, then the assertion follows immediately. For an Axiom A invariant set, Sigmund~\cite[Thm.1]{sigmund-72}
has proven that invariant measures supported by periodic orbits of $K$ are dense in the set of
invariant probability measures supported by $K$. In particular, any invariant measure can be approximated by an ergodic measure and, consequently, the result follows.
\end{proof}

Another special case where the necessary and sufficient conditions
coincide is when $K$ supports all populations except one. For discrete-time models,
this  case was considered by \cite{salceanu-smith-09}.

\begin{corollary}\label{cor:cod1}
Assume $x\mapsto A(x)$ is  twice continuously differentiable. Let $K\subset \prod_{i\ge 2} \Ci^+$ be a compact invariant set. Then
$K$ is robustly unsaturated if and only if $\lambda_1(\mu)>0$ for
all ergodic measures supported by $K$.
\end{corollary}

\begin{proof}
Since $K\subset \prod_{i\ge 2}\Ci^+$,
Proposition~\ref{prop:invasion2} implies that $\lambda_i(\mu)=0$ for
all ergodic measures $\mu$ supported by $K$ and $2\le i\le k$. The
ergodic decomposition theorem implies that $\lambda_i(\mu)=0$ for
all invariant measures $\mu$ supported by $K$ and $2\le i\le k$.
Therefore, for any invariant measure $\mu$ supported by $K$, \eqref{eq:condition} holds if
and only if $\lambda_1(\mu)>0$. The ergodic decomposition theorem implies
  $\lambda_1(\mu)>0$ for all invariant measures supported
by $K$ if and only if $\lambda_1(\mu)>0$ for all ergodic measures
supported by $K$.
\end{proof}

\section{Morse decompositions and robust permanence\label{sect:permanence}}

To state the sufficient condition for robust permanence, we use a characterization of permanence due to Garay~\cite{garay-89} and Hofbauer and
So~\cite{hofbauer-so-89} that involves Morse
decompositions of the boundary flow.  Conley~\cite{conley-78}
defined a collection of sets
$\{M_1,\dots, M_m\}$ to be a {\em Morse decomposition} for a compact
invariant set $K$ if \begin{itemize} \item $M_1,\dots,M_m$ are
pairwise disjoint, compact isolated invariant sets for the flow of \eqref{eq:main}
 restricted to $K$. \item For each $x\in K$ there are integers
$r=r(x)\le s=s(x)$ such that
 $\alpha(x)\subseteq M_r$ and $\omega(x)\subseteq M_s$.
 \item If
$r(x)=s(x)$, then $x\in M_r$. \end{itemize} Garay, Hofbauer and
So~\cite{garay-89,hofbauer-so-89} proved the following
characterization of permanence.

\begin{theorem}[Garay, Hofbauer-So]\label{thm:permanence} If
$\{M_1,\dots, M_m\}$ is a Morse decomposition for $\partial \C$,
then (\ref{eq:main}) is permanent if and only if each of the
components $M_i$ are unsaturated.
\end{theorem}

Theorem~\ref{thm:unsaturated} and \ref{thm:permanence} imply the
following result:

\begin{theorem}\label{thm:rp}
If $\{M_1,\dots, M_m\}$ is a Morse decomposition for $\partial \C$
and (\ref{eq:condition}) holds for each of the components of the
Morse decomposition, then \eqref{eq:main} is robustly permanent.
Conversely, if $x\mapsto A(x)$ is  twice continuously differentiable and \eqref{eq:condition} is violated by an ergodic measure supported by
one of the components of the Morse decomposition, then \eqref{eq:main} is not robustly permanent. 
\end{theorem}

Theorem~\ref{thm:rp} in conjunction with Corollaries~\ref{cor:axiomA} and \ref{cor:cod1} yield a characterization of robust permanence for a class of structured models. 

\begin{corollary} 
Let $\{M_1,\dots, M_m\}$ is a Morse decomposition for $\partial \C$. Assume
$x\mapsto A(x)$ is  twice continuously differentiable and for each Morse component $M_i$ one of the following assertions hold
\begin{itemize}
\item $M_i$ is Axiom A,
\item $M_i$ is uniquely ergodic, or
\item there exists $j\in \{1,\dots,k\}$ such that $M_i\subset \{ x\in \C: x_l\gg 0 \mbox{ for }l\neq j\}$.
\end{itemize}
Then \eqref{eq:main} is robustly permanent if and only if \eqref{eq:condition} holds for all ergodic measures supported by $\cup_i M_i$.
\end{corollary}

Our results also provide a structured analogue for characterizing  Òtotally permanent systemsÓ ~\cite{jmaa-02}. For an ergodic probability measure $\mu$, define $\supp(\mu)$ to be the subset $I\subset \{1,\dots, k\}$ such that
$\mu(\prod_{i\in I}\C_+)=1$.


\begin{corollary} The following statements are equivalent:
\begin{itemize}
\item for all ergodic probability measures $\mu$ with support in $\partial \C$,
\[
\lambda_i(\mu)>0\mbox{ for all }i\in \{1,\dots, k\}\setminus \supp(\mu)
\]
\item \eqref{eq:main} and all of its subsystems are robustly permanent.
\end{itemize}
\end{corollary}

\begin{proof}
Suppose the first statement holds. Then unsaturated condition given by \eqref{eq:condition2} with $p=(1,\dots,1)$ and Theorem~\ref{thm:rp} implies that \eqref{eq:main} and all its subsystems are robustly permanent. The other direction follows from the necessary condition for robust permanence in Theorem~\ref{thm:rp} and the second assertion of Proposition~\ref{prop:invasion2}. \end{proof}

To illustrate the broad applicability of Theorem~\ref{thm:rp}, we develop applications to spatially structured models in section~\ref{sect:patch} and to a disease model and a gene network model in section~\ref{sect:applications}.

\section{Applications: Patch models\label{sect:patch}}

A fundamental application of our results are to spatially structured models, with $k$ species dispersing between $m$ patches~\cite{levin-74,yodzis-78,czaran-98}:
\begin{equation}\label{eq:patch}
\frac{d x_i^j}{dt} = x_i^j f_i^j(x^j) + \sum_l d_i^{jl}x_i^l - e_i^j x_i^j  \qquad i=1, \dots, k \quad j = 1, \dots, m
\end{equation}
where $x_i^j$ denotes the density of species $i$ in patch $j$, $x^j=(x_1^j,\dots,x_k^j)$ is the vector of species densities in patch $j$, $f_i^j$ is the per-capita growth rate of species $i$ in patch $j$, $d_i^{jl} \geq 0$ is the dispersal rate for species $i$ from patch $l$ into patch $j$, and $e_i^j$ is the emigration rate of species $i$ out of patch $j$. Hence $e_i^j \geq \sum_l d_i^{lj}$. We assume that the matrices $(d_i^{jl})_{j,l}$ are irreducible for each species $i$.
It is then easy to write \eqref{eq:patch} in the form \eqref{eq:main}, with $x_i = (x_i^1,\dots,x_i^m)'$ where $'$ denotes transpose.

For a single species \eqref{eq:patch} generates a monotone flow. Under mild assumptions, e.g., each $f_i^j$ is decreasing and negative for large densities, there is a globally stable equilibrium, see, e.g., \cite{freedman-takeuchi-89} and \cite[sect. 5.4]{hofbauer-90}.
This equilibrium will be positive if the invasion rate at the origin  $\lambda(\0)$, which is given by the leading eigenvalue (stability modulus) of the matrix $(d^{jl} + f^j(\0) - e^j) $,  is positive.

\subsection{Two species}

For \emph{two competing species},  e.g., each $f_i^j$ is decreasing with respect to both species and is negative at large densities,  the dynamics are still monotone, and hence almost all orbits converge to an equilibrium.  Robust permanence requires two conditions. First, $\lambda_i(\0)>0$ for both species in which case there are two single species equilibria $\mathbf{E}_1$ and $\mathbf{E}_2$.  Second, the invasion rates $\lambda_2(\mathbf{E}_1)$ and $\lambda_1(\mathbf{E_2})$  are positive. Whether there is a unique positive globally stable equilibrium for the robustly permanent system depends in a delicate way on the system parameters, see \cite{hofbauer-etal-96}. For slowly dispersing populations (i.e. $0<e_i^j\ll 1$ for all $i,j$), the robust permanence condition is particularly straight forward to verify. Let $E_i^j$ be the largest solution to $x_i^jf_i^j(x_i^j)=0$ i.e. the equilibrium attained by species $i$ in patch $j$ when there is no dispersal and no competitors. For slowly dispersing populations,  $\mathbf{E}_1$ is close to the product $(E_1^1,\dots, E_1^m)$, and the matrix $A^2(\mathbf{E}_1)$ is close to the diagonal matrix $\diag(f_2^j(E_1^j))_{j = 1}^m$. Hence the invasion rate $\lambda_2(\mathbf{E_1})$ is close to $\max_j f_2^j(E_1^j)$.
So we obtain: The system with two slowly dispersing competing species  is robustly permanent if and only if
\begin{equation}
\max_j f_2^j(E_1^j) > 0 \quad \text{and} \quad \max_j f_1^j(E_2^j) > 0
\end{equation}
Intuitively, robust permanence  requires for each species there
is at least one patch where it can persist.

For \emph{predator--prey systems} the situation is even easier. Since (under mild assumptions on the prey dynamics) the Morse decomposition of  $\partial \C$ is given by two equilibria $\{\0, \mathbf{E}_1\}$, the spatial predator--prey system is robustly permanent, if the predator can invade at the prey equilibrium: $\lambda_2(\mathbf{E}_1)> 0$. However, the global dynamics are likely to be more complicated.
For instance, if in each patch there is a globally stable limit cycle (e.g., \cite{cheng-81,kuang-freedman-88,hasik-00}) and dispersal is sufficiently slow (i.e. $0<e_i^j\ll 1$ for all $i,j$), our results about robust permanence and the theory of normally hyperbolic manifolds~\cite{hirsch-pugh-shub-77} imply that  there is a positive $m$ dimensional torus, which attracts almost all orbits in $\C$, whenever $\max_j f_2^j(E_1^j)>0$.

\subsection{ Rock-paper-scissors dynamics}

The Lotka-Volterra model of rock-paper-scissor dynamics is a simple model that is used as prototype for understanding intransitive ecological outcomes~\cite{may-leonard-75,schreiber-killingback-preprint}. Here, a simple spatial version of this dynamic is given by
\begin{eqnarray*}
\frac{d  x_1^j}{dt} &=& x_1^j (1-x^j_1-\beta^j x^j_2 - \alpha^j x^j_3)+ \sum_k d_1^{jk}x_1^k - e_1^j x_1^j\\
\frac{d  x_2^j }{dt}&=& x_2^j (1-\alpha^j x^j_1- x^j_2 - \beta^j x^j_3)+ \sum_k d_2^{jk}x_2^k - e_2^j x_2^j\\
\frac{d  x_3^j}{dt} &=& x_i^j (1-\beta^j x^j_1-\alpha^j x^j_2 - x^j_3)+ \sum_k d_3^{jk}x_3^k - e_3^j x_3^j\\
\end{eqnarray*}
where $\alpha^j\in (0,1)$ and $\beta^j>1$ for all $j$. A more general version of this dynamic is presented in \cite{schreiber-killingback-preprint}. Under the assumption that there is no cost to dispersal (i.e. $e_i^j =\sum_l d_i^{lj}$ for all $i,j$), the maximal invariant set in $\partial \C$ consists of the origin $\0$ and a heteroclinic cycle connecting positive single species equilibria $\mathbf{E_1}$, $\mathbf{E_2}$, and $\mathbf{E_3}$. For slowly dispersing populations ($0<e_i^j \ll 1$ for all $i,j$), $\mathbf{E}_1$ is close to $(x_1,x_2,x_3)=(1,\dots,1,0,\dots,0,0,\dots,0)$ and the invasion rates $\lambda_2(\mathbf{E}_1)$ and $\lambda_3(\mathbf{E}_1)$ are close to $\max_j 1-\alpha^j$ and $\max_j 1-\beta^j$. For the other equilibria $\mathbf{E}_i$, the invasion rates of the missing species are also given by $\max_j 1-\alpha^j$ and $\max_j 1-\beta^j$. Consider the Morse decomposition of $\partial \C$ given by $\0$ and the heteroclinic cycle. Since $\lambda_i(\0)=1>0$ for $i=1,2,3$, $\0$ is robustly unsaturated. An algebraic computation reveals that \eqref{eq:condition} holds for all invariant measures supported by the heteroclinic cycle (i.e. all convex combinations of the Dirac measures supported by the equilibria) if and only if
\[
\max_j (1-\alpha^j) > - \max_j (1-\beta^j)
\] 
Equivalently,
\[
2> \min_j \alpha^j + \min_j \beta^j
\]
In particular, even if the heteroclinic cycle is attracting for each patch when the system is uncoupled (i.e. $2<\alpha^j+\beta^j$ for all $j$), it can be repelling for the weakly coupled system.

\subsection{Three species Lotka Volterra in spatially homogenous environments}

Consider Lotka--Volterra dynamics in a spatially homogenous environment:
\begin{equation}\label{LVD}
\frac{dx_i^j}{dt} = x_i^j (r_i - \sum_s a_{is} x_s^j) + \sum_l d_i^{jl}x_i^l - e_i^j x_i^j ,
\qquad i=1,\dots,k; \quad j = 1, \dots, m
\end{equation}
We assume that dissipativity can be shown by a linear Liapunov function:
There are $c_i > 0$ such that $\sum_i c_i a_{is}x_ix_s < 0$ holds for all $x \in \R^k_+ \setminus \{\0\}$.
Then the weighted sum of densities across all patches $\sum_{i,j} c_i x_i^j$ is decreasing for large densitities  and hence \eqref{LVD} generates a dissipative semiflow with global attractor $\Gamma$.
We further assume that for each patch $j$ and each species $i$ the immigration rate equals the emigration rate, i.e.,
\begin{equation}\label{eq:im=em}
e_i^j = \sum_l d_i^{jl}.
\end{equation}
This guarantees that the set $\H = \{x \in \C: x_i^j = x_i^l  \  \forall i,j,l\}$ of spatially homogenous states  is forward invariant under \eqref{LVD}.

Now consider  the Lotka--Volterra dynamics without spatial structure
\begin{equation}\label{LV}
\frac{d x_i}{dt} = x_i (r_i - \sum_s a_{is} x_s),
\qquad i=1,\dots, k
\end{equation}
It has been shown \cite[ch. 16.1, 16.2]{hofbauer-sigmund-98} that for $k=3$, \eqref{LV}  is robustly permanent (robust meaning here within the class of Lotka--Volterra systems) if and only if  all equilibria on the boundary of $\R^3_+$ are unsaturated and whenever there is a heteroclinic cycle connecting the one species equilibria then this cycle is repelling. Equivalently the following four conditions hold.
\begin{enumerate}
\item[(i)] there exists an interior equilibrium $\hat x$ (i.e. $A\hat x = r$ with $\hat x \gg 0$);
\item[(ii)] $\det (-A) > 0$;
\item[(iii)] the 2 species subsystems  are not bistable competition systems;
\item[(iv)] if there is a heteroclinic cycle between the  one species equilibria, say $\mathbf{E}_1 \to \mathbf{E}_2 \to \mathbf{E}_3  \to \mathbf{E}_1$ then the following inequality holds
\begin{equation}\label{eq:hetcy}
\lambda_2(\mathbf{E}_1)\lambda_3(\mathbf{E}_2)\lambda_1(\mathbf{E}_3) > |\lambda_3(\mathbf{E}_1)\lambda_1(\mathbf{E}_2)\lambda_2(\mathbf{E}_3)|
\end{equation}
\end{enumerate}
Furthermore, the proof in \cite[ch. 16.1, 16.2]{hofbauer-sigmund-98} shows that
 the boundary flow of \eqref{LV} has a simple Morse decomposition:
if (iv) applies then there are two Morse sets, the origin $\0$ is a repeller, and the heteroclinic cycle is the dual attractor (within $\partial \R^3_+$); if (iv) does not apply then the finitely many boundary equilibria form a Morse decomposition of the boundary flow.
 
\begin{theorem}\label{thm:lvd}
Under the above assumptions, \eqref{LVD} for $k=3$ is robustly permanent if and only if (i), (ii), (iii) and (iv) hold.
\end{theorem}
\begin{proof}
Suppose \eqref{LVD} is robustly permanent.
Since the nonspatial system  \eqref{LV} is the restriction of \eqref{LVD} to the invariant subspace $\H$, \eqref{LV} must be robustly permanent and hence (i) - (iv) hold.

Conversely, suppose that (i) - (iv) hold. We show that  the above Morse decomposition of the boundary flow of \eqref{LV} (now in $\H$) is also a Morse decomposition for \eqref{LVD} restricted to the maximal invariant subset of $\partial \C$ which  is contained in the spatial two species  subsystems of \eqref{LVD}.
For this we show that in each 2 species system the orbits  of \eqref{LVD} converge to $\H$ (For the one species subsystems this is obvious.)
If a two species subsystem of \eqref{LV} has an internal equilibrium, say $\mathbf{E}_{12}$ then by assumption  (ii), this is not a saddle but a  sink, and hence $a_{11}a_{22} > a_{12}a_{21}$. This implies that the $1$--$2$ submatrix of $A$ is VL-stable \cite[ch. 15.3]{hofbauer-sigmund-98}, and
\cite{hastings-78} implies that $\mathbf{E}_{12}$ is the global attractor for the spatially structured system \eqref{LVD} restricted to the first two species.

If a two species subsystem of \eqref{LV} has no internal equilibrium, then a one species equilibrium say $\mathbf{E_1}$ is the global attractor.
Then Lemma~\ref{l:kishimoto} (see below) shows the corresponding result for \eqref{LVD}.
A similar proof applies to the case where  the origin $\0$ is the global attractor of  a two species subsystem.

Finally observe that under the above assumptions \eqref{LVD}--\eqref{eq:im=em}, the invasion rates $\lambda_i$ at spatially homogenous boundary equilibria are the same for \eqref{LVD} and \eqref{LV}. Since all boundary equilibria are unsaturated for \eqref{LV}, they are for \eqref{LVD} as well. The same applies in case (iv) to the heteroclinic cycle. Note that  the condition \eqref{eq:hetcy} on the  eigenvalues 
is equivalent to the fact that all invariant measures supported on the heteroclinic cycle are unsaturated, see \cite[Ex. 4.5]{garay-hofbauer-03}.

\begin{lemma}\label{l:kishimoto}
Suppose, in a two species Lotka--Volterra  system  \eqref{LV} ($k=2$) all interior orbits converge to the one species equilibrium $\mathbf{E}_1$.
Then all orbits of the spatial version \eqref{LVD} ($k=2$) converge to the spatially homogeneous one species equilibrium $\mathbf{E}_1 \in \H$.
\end{lemma}

\begin{proof}
If the system is of predator--prey or mutualistic type then Hasting's \cite{hastings-78} result applies again. So
let us assume that the local interaction  is given by the competition system
\begin{equation}\label{LV2}
\frac{dx_1}{dt} = r_1 x_1 (1 - x_1 - \alpha x_2) \qquad
\frac{d x_2}{dt} = r_2 x_2 (1 - \beta x_1 - x_2)
\end{equation}
with $0 < \alpha <1 < \beta$. In this case,  Kazuo Kishimoto (letter to JH, Oct 1987) has given the following argument.
There exists a family of forward invariant rectangles contracting to $\mathbf{E}_1$. This shows that all interior solutions of \eqref{LVD} ($k=2$) converge to $\mathbf{E}_1$.
\end{proof}
\end{proof}

Theorem \ref{thm:lvd} shows that a spatial network of identical patches  is permanent, if the within patch dynamics is permanent, and three species Lotka-Volterra dynamics.
It is not clear how to extend this result to more than three species. The crucial step in the three species case is that permanence precludes
bistable two species subsystems. The spatial version of a bistable two species system allows plenty of stable equilibria outside $\H$, see \cite{levin-74}. However, permanent four species Lotka Volterra systems may have two species bistable subsystems  \cite[ch.16.4]{hofbauer-sigmund-98}.
So there may be invariant sets on the boundary outside $\H$ which need to be unsaturated to make \eqref{LVD} permanent.

\section{More Applications\label{sect:applications}}

\subsection{Disease dynamics with density-dependent demography\label{sect:applications:disease}}
Gao and Heth\-cote~\cite{gao-hethcote-92} introduced a model of disease dynamics with density-dependent demo\-graphy. Here we consider a variation of their model in which a population of size $P$ has a per-capita birth rate $b(P)$ and per-capita death rate $d$. To ensure that the population persists in the absence of the disease, we assume that there exists $K>0$ such that $b(K)=d$. To describe the disease dynamics, let $S$ be the number of individuals susceptible to the disease, $I$ the number of infected individuals, and $R$ the number of individuals that have recovered from the disease. We assume that $P=S+I+R$. If the disease transmission is asymptotic, $\beta$ is the contact rate between susceptible and infected individuals, $\gamma$ is the rate at which individuals recover from the disease,  and $m$ the mortality rate due to the disease, then the population dynamics are given by
\begin{eqnarray*}
\frac{dP}{dt}&=& b(P)P- dP - m I \\
\frac{dI}{dt}&=& \beta (P-I-R) \frac{I}{\epsilon+P} - (\gamma+d+m)I\\
\frac{dR}{dt}&=& \gamma I - d R
\end{eqnarray*}
where $\epsilon>0$ is a constant. It is useful to introduce a change of coordinates in which $y=\frac{I}{P}$ and $z=\frac{R}{P}$:
\begin{eqnarray*}
\frac{dP}{dt}&=& \left( b(P)- d - my\right) P\\
\frac{dy}{dt}&=& \left( \beta(1-y-z)\frac{P}{\epsilon+P}-\gamma-m(1-y)-b(P)\right) y\\
\frac{dz}{dt}&=& \gamma y -\left(b(P)-my\right)z
\end{eqnarray*}
Setting $x_1=P$ and $x_2=(y,z)'$ where $'$ denotes transpose yields a structured model where $A_1(x)=b-d- my$ and
\[
A_2(x)=\begin{pmatrix}\beta(1-y-z)\frac{P}{\epsilon+P}-\gamma-m(1-y)-b(P) & 0 \cr \gamma & my-b(P)\end{pmatrix}
\]
Since $\{\0, (K,0,0)\}$ is a Morse decomposition of the boundary dynamics, Theorem~\ref{thm:rp} and Remark~\ref{remark}  imply that this model is robustly permanent if and only if $b(0)>d$ and $\beta\frac{K}{\epsilon+K}>\gamma+m+d$. In particular, for $\epsilon>0$ sufficiently small, one requires that the basic reproductive number $\frac{\beta}{\gamma+m+d}$ is greater than one.

Remarkably, this same criterion determines robust permanence of models with significantly more complicated boundary dynamics. For example, suppose the focal population is a predator species. If the prey has abundance $N$ and exhibits logistic dynamics, $f(N)$ is the per-capita predator consumption rate of the prey, and $b(N)$ is the per-capita reproductive rate of the predator, then the dynamics become
\begin{eqnarray*}
\frac{dN}{dt}&=& rN(1-N/K)- f(N)P\\
\frac{dP}{dt}&=& \left( b(N)- d - my\right) P\\
\frac{dy}{dt}&=& \left( \beta(1-y-z)\frac{P}{\epsilon+P}-\gamma-m(1-y)-b(N)\right) y\\
\frac{dz}{dt}&=& \gamma y -\left(b(N)-my\right)z
\end{eqnarray*}
where $r>0$ and $K>0$, and $f(N),b(N)$ satisfy $f(0)=b(0)=0$. Define $x_1=N$, $x_2=P$, $x_3=(y,z)'$. Let us assume that the  disease-free community (i.e. $y=z=0$) is permanent (i.e. $b(K)>d$) and let $A$ be the global attractor in the interior of the $P$--$N$ plane. Under these assumptions, a Morse decomposition of the boundary dynamics is given by $(x_1,x_2,x_3)=(0,0,0)$, $(x_1,x_2,x_3)=(K,0,0)$, and $A$. Since the equilibria $(0,0,0)$ and $(K,0,0)$ are unsaturated, it remains to characterize $\lambda_3(\mu)$ for any invariant measure $\mu$ supported by $A$. Proposition~\ref{prop:invasion2} and Remark~\ref{remark} imply
\[
\lambda_2(\mu)=\int_\C (b(x_1)-d) \, d\mu(x)=0
\]
Hence $\int_C b(x_1)\,d\mu(x)=d$ and
\[
\lambda_3(\mu)=\beta\int \frac{x_2}{\epsilon+x_2}d\mu(x)-\gamma -m -d
\]
Theorem~\ref{thm:rp} and Remark~\ref{remark} imply that this structured model is robustly permanent if $\frac{\beta}{\gamma+m+d}>1$ and $\epsilon>0$ is sufficiently small.

\subsection{Gene networks}

 The repressilator is an oscillatory gene network based on three (or more generally an odd number of) transciptional repressors. Such a system has been genetically engineered in E-coli by Elowitz and Leibler \cite{elowitz-leibler-00}. Mathematical models go back to \cite{banks-mahaffy-78, smith-87}. They all involve concentrations of  proteins and mRNAs. A modified model with auto-activation was suggested in \cite{muller-etal-06}
\begin{eqnarray}\label{RA}
\dot x_i &=& \beta(y_i - x_i) \\
\dot y_i &= &\alpha F(x)_i - y_i
\end{eqnarray}
with
\begin{equation}
F(x)_i = x_i g(x_i,x_{i-1}) =
\frac{x_i}{1 + x_i + \rho x_{i-1} + \kappa\rho  x_i x_{i-1}}
\end{equation}
Here $x_i$ and $y_i$ are normalized concentrations of  proteins and mRNAs belonging to gene $i$, $\rho$ is the strength of repression and $\kappa$ a cooperativity parameter. This is a structured system of type \eqref{eq:main}
\begin{equation} \label{RAc}
\begin{pmatrix} \dot{x_i} \\ \dot{y_i} \end{pmatrix} =
\begin{pmatrix} - \beta &  \beta \\
\alpha g(x_i,x_{i-1}) & - 1 \end{pmatrix}
\begin{pmatrix} x_i \\ y_i \end{pmatrix}
\end{equation}

In \cite[Thm 3]{muller-etal-06},  it was shown that for $n= 3$,   $\alpha>1$ and $\rho>1$,
system \eqref{RA} has a heteroclinic cycle connecting the $3$ single gene equilibria (similar to the rock--scissors--paper dynamics of section 5.2). Moreover
\begin{itemize}
\item the system is permanent, i.e., the boundary of $\R^{6}_+$ is repelling,
if $\lambda + \mu> 0$, 
\item this heteroclinic cycle is asymptotically stable if
$\lambda +  \mu < 0$
\end{itemize}
where $\lambda > 0$ and $\mu < 0$ are the invasion rates of the two missing genes, i.e., the leading eigenvalues of the $2 \times 2$ matrix in \eqref{RAc} evaluated at the appropriate single gene equilibria (for explicit expressions see \cite[eqs.\ (127, 128)]{muller-etal-06}.)

Theorem \ref{thm:rp} implies the permanence condition implies  the stronger conclusion of
robust permanence. A similar result holds for  $n$ odd, but the heteroclinic cycle is then actually a heteroclinic network.

\bibliography{../../seb}

\end{document}